\documentclass[conference,letterpaper]{IEEEtran}

\IEEEoverridecommandlockouts

\usepackage[utf8]{inputenc} 
\usepackage{liustyle}
\usepackage[T1]{fontenc}
\usepackage{url}
\usepackage{ifthen}
\usepackage{cite}

\usepackage{setspace}
\usepackage{svg}

\begin{document}
\title{Capacity-achieving sparse superposition codes with spatially coupled VAMP decoder
\thanks{The first two authors contributed equally to this work and ${\dagger}$ marked the corresponding author.}
} 

\author{%
 \IEEEauthorblockN{Yuhao Liu, Teng Fu, Jie Fan, Panpan Niu, Chaowen Deng and Zhongyi Huang$^{\dagger}$}
\IEEEauthorblockA{Department of Mathematical Sciences, Tsinghua University, Beijing, China\\
                   Email:  \{yh-liu21, fut21, fanj21, npp21 and dcw21\}@mails.tsinghua.edu.cn, zhongyih@tsinghua.edu.cn} 
}


\maketitle


\begin{abstract}
     Sparse superposition (SS) codes provide an efficient communication scheme over the Gaussian channel, utilizing the vector approximate message passing (VAMP) decoder for rotational invariant design matrices \cite{hou2022sparse}. Previous work has established that the VAMP decoder for SS achieves Shannon capacity when the design matrix satisfies a specific spectral criterion and exponential decay power allocation is used \cite{xu2023capacity}.
    In this work, we propose a spatially coupled VAMP (SC-VAMP) decoder for SS with spatially coupled design matrices. Based on state evolution (SE) analysis, we demonstrate that the SC-VAMP decoder is capacity-achieving when the design matrices satisfy the spectra criterion.  Empirically, we show that the SC-VAMP decoder outperforms the VAMP decoder with exponential decay power allocation, achieving a lower section error rate. All codes are available on {https://github.com/yztfu/SC-VAMP-for-Superposition-Code.git}.
\end{abstract}

\section{Introduction}
Sparse superposition (SS) codes, also called sparse regression codes, were originally proposed by Barron and Joseph for reliable communication over the additive white Gaussian noise channel (AWGNC) \cite{barron2010toward, barron2011analysis, joseph2012least}. While the scheme was known to be capacity-achieving for error correction with (intractable) maximum likelihood decoder \cite{venkataramanan2019sparse}, its computational complexity grows exponentially with code lengths. To mitigate this, the adaptive successive hard-decision decoder, a polynomial-time decoder, was proposed in \cite{joseph2013fast, cho2013approximate}. 
However,  both theoretical and empirical performances remain suboptimal: the asymptotic results do not extend to practical finite section lengths, and the empirical section error rate at feasible code lengths remains high for rates near capacity.

An iterative soft-decision decoder for SS based on approximate message passing (AMP) algorithm \cite{donoho2009message, bayati2011dynamics}, with polynomial computational complexity, was proposed in  \cite{barbier2014replica}, showing improved error performance at finite code lengths. The AMP decoder for SS, with power allocation (PA)  or spatially coupled (SC)  Gaussian design matrix, was shown asymptotically capacity-achieving through non-rigorous replica  symmetric potential analysis \cite{barbier2017approximate, barbier2016proof}. This was later fully rigorously proven by Rush et al., using conditional techniques, establishing that the AMP decoder achieves Shannon capacity over AWGNC and provides non-asymptotic section error rate bounds\cite{rush2017capacity, rush2018error,rush2021capacity}. 
Similar analysis can be extended to show that the AMP decoder for SS with SC Gaussian design matrix, based on generalized AMP (GAMP) algorithm \cite{rangan2011generalized,feng2022unifying}, is capacity-achieving over memoryless channels \cite{barbier2019universal, liu2024error}.

In practice, non-Gaussian design matrix,  such as Hadamard and discrete cosine transform (DCT) matrices \cite{greig2017techniques}, which are well approximated by orthogonally invariant matrices, can reduce the computational complexity for AMP decoder. However, the analysis in \cite{rush2017capacity} does not hold for orthogonally invariant matrices. To address this, \cite{hou2022sparse} proposed a decoder for SS with more structured design matrices, namely, right
orthogonally invariant matrices, based on  vector approximate message passing (VAMP) algorithm \cite{rangan2019vector} (also known as orthogonal approximate message passing, OAMP \cite{ma2017orthogonal}). They conjectured that VAMP decoder for SS is capacity-achieving if the design matrices satisfy the spectra criterion \cite{hou2022sparse}. Subsequently, \cite{xu2023capacity} rigorously proved that VAMP decoder for SS equipped with exponential decay PA is capacity-achieving under this spectra criterion and provided a non-asymptotic performance characterization of section error rate.


However, the empirical performance of SS with exponential decay PA remains unsatisfactory. Inspired by the superior practical performance of SS with spatially coupled (SC) Gaussian design matrices \cite{barbier2017approximate, hsieh2018spatially}, we introduce, for the first time, a spatially coupled VAMP (SC-VAMP) decoder for SS with spatially coupled rotational invariant design matrices. We first derive the SC-VAMP decoder using the factor graph  \cite{kschischang2001factor} and the expectation consistent (EC) algorithm \cite{minka2001family, seeger2005expectation}. Next, we demonstrate that the SC-VAMP decoder with design matrices satisfying the spectral criterion is capacity-achieving over AWGNC by state evolution (SE) analysis in the large-section limit. Finally, we empirically show that the SC-VAMP decoder outperforms the VAMP decoder with exponential decay PA for SS \cite{xu2023capacity}, achieving a lower section error rate.

\textit{Notation}: 
We denote by $[n]$ the set $\{1, 2, \dots, n\}$ for a positive integer $n$. For a matrix $\bA \in \RR^{m \times m}$, we define $\langle \bA \rangle$ as $m^{-1} \operatorname{Tr}(\bA)$, where $\operatorname{Tr}(\bA)$ represents the trace of $\bA$. For a set $S$, $|S|$ denotes its cardinality. If a vector $\vec{\bp}$ is partitioned into $\sfC$ blocks, $\vec{\bp} = [\bp_1, \dots, \bp_{\sfC}]$, we use $\vec{\bp}[\sfc]$ to refer to the $\sfc$-th block $\bp_{\sfc}$. Additionally, we define $\vec{\bx} = \text{concate}\{\bx_{\sfc}| \sfc \in S\}$ as the vector $\vec{\bx}$ obtained by concatenating $\bx_{\sfc}$ for $\sfc \in S$  in ascending order of the indices.


\section{Spatially coupled SS construction}

In the context of SS codes, the \textit{message} vector $\bx = [\bx_{\text{sec}(1)}, \cdots, \bx_{\text{sec}(L)}] \in \RR^N$ consists of $L$ sections, each with $B$ entries, where $B$ is the \textit{section size} and $N=LB$. The set $\text{sec}(l):=\{(l-1)B+1, \cdots, lB\}$ collects the indices in the $l$-th section, with $\bx_{\text{sec}(l)} = [x_{(l-1)B+1}, \cdots, x_{lB}]$. Each section $\bx_{\text{sec}(l)}, l \in \{1, \cdots, L\}$ contains a single non-zero component, equal to one, whose position encodes the symbol to be transmitted. In the spatially coupled SS (SC-SS), the message is divided into $\sfC$ blocks equally, each containing $L/\sfC$ sections, such that $\bx = [\bx_1, \cdots, \bx_{\sfC}]$. Each block $\bx_{\sfc}, \sfc \in [\sfC]$ concatenates the $l_{\sfc}$-th section to $r_{\sfc}$-th section, i.e., $\bx_{\sfc} = [\bx_{\text{sec}(l_{\sfc})}, \cdots, \bx_{\text{sec}(r_{\sfc})}]$ where $l_{\sfc} = (\sfc-1)L/ \sfC+1$ and $r_{\sfc} = \sfc L/\sfC$. A \textit{codeword} $\bz = [\bz_1, \cdots, \bz_{\sfR}]\in \RR^M$ is divided into $\sfR$ blocks equally, with each block $\bz_{\sfr}$ generated from a fixed \textit{design matrix} $\bA_{\sfr} \in \RR^{M_{\sfr} \times N_{\sfr}}$, such that $\bz_{\sfr} = \bA_{\sfr} \vec{\bx}_{\sfr}$, for $\sfr \in [\sfR]$. The vector $\vec{\bx}_{\sfr}$ concatenates the  blocks $\bx_{\sfc}$ encoded by $\bA_{\sfr}$, with the indices of these blocks collected in $W_{\sfr}$, such that $M_{\sfr} = M /\sfR$ and $N_{\sfr} = N|W_{\sfr}| / \sfC$ and $\vec{\bx}_{\sfr} = \text{concate}(\{\bx_{\sfc}|\sfc \in W_{\sfr}\})$.
We consider random codes generated by $\bA_{\sfr}$, where $\bA_{\sfr}$ is drawn from rotational invariant ensembles, i.e., the singular value decomposition $\bA_{\sfr} = \bU_{\sfr} \sqrt{\bS_{\sfr}} \bV_{\sfr}^{\mathsf{T}}$ involves orthogonal bases $\bU_{\sfr}$ and $\bV_{\sfr}$ sampled uniformly from the orthogonal group $\cO(M_{\sfr})$ and $\cO(N_{\sfr})$, respectively.  
The diagonal matrix $\bS_{\sfr}$ contains the squares of the singular values of $\bA_{\sfr}$, denoted $(S_{\sfr i})_{i\leq N_{\sfr}}$, on its main diagonal. The empirical distribution $N_{\sfr}^{-1} \sum_{i\leq N_{\sfr}} \delta_{S_{\sfr i}}$ weakly converges to a well-defined, compactly supported probability density function as $N_{\sfr}, M_{\sfr} \to \infty$ (not necessarily proportionally). We denote  the aspect ratio of $\bA_{\sfr}$ as $\alpha_{\sfr} = M_{\sfr}/N_{\sfr}$ ($0<\alpha_{\sfr}<1$ generally) and $\rho_{\sfr} = (1-\alpha_{\sfr}) \delta_0 + \alpha_{\sfr} \rho_{\text{supp},\sfr}$ the limit spectral density of $B^{-1} \bA_{\sfr}^{\mathsf{T}} \bA_{\sfr}$, with $\rho_{\text{supp}, \sfr}$ the limit spectral density on $\RR^{+}$ as $L \to \infty$. The overall (design) rate of the SC-SS is $R_{\text{all}} = L \log(B) / M = N\log(B) / (MB)$, and the rate associated with $\boldsymbol{z}_{\sfr}$ is $R_{\sfr} = N_{\sfr} \log(B) / (M_{\sfr}B) = \sfR |W_{\sfr}| R_{\text{all}} / \sfC$. Thus, the code is fully specified by the tuple $(M, R_{\text{all}}, B, \sfR, \sfC, \{W_{\sfr}, \rho_{\sfr}\}_{\sfr=1}^{\sfR})$. In this paper, we set $\sfC = \Gamma, \sfR = \Gamma+W-1$ and $\vartheta = (\Gamma+W-1)/\Gamma$, considering the following $\{W_{\sfr}\}_{\sfr=1}^{\sfR}$ inspired by the coupling structure shown in \cite{barbier2017approximate,barbier2019universal,rush2021capacity,takeuchi2023orthogonal,cobo2024bayes}
\begin{equation*}
W_{\sfr} = 
\begin{cases}
\{1, 2, \dots, \sfr\}, & 1 \leq \sfr < W, \\
\{\sfr-W+1, \dots, \sfr\}, & W \leq \sfr \leq \Gamma, \\
\{\sfr+1-W, \dots, \Gamma\}, & \Gamma < \sfr \leq \Gamma+W-1.
\end{cases}
\end{equation*}
With this structure, a SC-SS can be defined by the tuple $(M, R_{\text{all}}, B, \Gamma, W, \{\rho_{\sfr}\}_{\sfr=1}^{\Gamma+W-1})$. For example, a SC-SS with parameters $\Gamma = 5, W=2$ is shown in Fig.~\ref{fig:SC_structure}. For the sake of clarity and convenience, we will interchangeably use $\sfC$ and $\Gamma$, as well as $\sfR$ and $\Gamma + W - 1$, in the following context.
\begin{figure}
    \centering
    \includegraphics[width=1.05\linewidth]{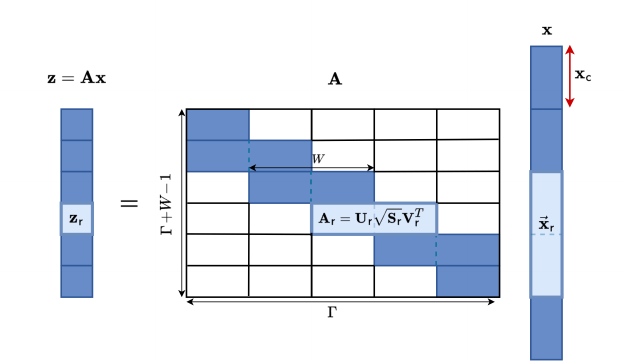}
    \vspace{-5mm}
    \caption{A SC-SS with spatial coupling parameters $\Gamma = 5$ and $W = 2$, thus $\sfR = 6$, $\sfC = 5$ and $\vartheta = 6/5$. Each $\bx_{\sfc}$ consists of $L/\Gamma$ for $\sfc \in [\sfC]$, and each $\vec{\bx}_{\sfr}$ consists of  $|W_{\sfr}| L / \Gamma$ sections for $\sfr \in [\sfR]$. Each design matrix $\bA_{\sfr}$ contains $M/(\Gamma+W-1)$ rows and $N|W_{\sfr}|/ \Gamma$ columns for $\sfr \in [\sfR]$.}
    \label{fig:SC_structure}
    \vspace{-6mm}
\end{figure}

We enforce the power constraint $\|\bz\|_2^2/M = 1+ o_L(1)$ by tuning $\bA_{\sfr}$'s spectrum to satisfy $\sum_{\sfr=1}^{\sfR} \int \lambda \rho_{\text{supp}, \sfr}(\lambda) \mathrm{d} \lambda = \sfR$ in the large $L$ limit. In this paper, we adopt the simplified setting that $\int \lambda \rho_{\text{supp}, \sfr} (\lambda) \mathrm{d} \lambda = 1$, ensuring $\| \bz_{\sfr}\|_2^2/M_{\sfr} = 1 + o_{L}(1)$ for all $\sfr \in [\sfR]$. Codewords are transmitted through an $\text{AWGNC}$, i.e., the received corrupted codewords are $\by = \bz + \bn$ with \iid $n_i \sim \cN(0, \sigma^2), i\leq M$, so that the signal-to-noise ratio is $\snr = \sigma^{-2}$. Then Shannon capacity of AWGNC is $\mathcal{C} = \log (1 + \snr) / 2$. 
In communication systems, one usually prefers right rotational invariant ensembles $\bA_{\sfr} = \sqrt{\bS_{\sfr}} \bV_{\sfr}^{\top}$. The choice
does not affect the theoretical analysis, since we can left-multiply $\by_{\sfr}$ by $\operatorname{diag} ( \{\bU_{\sfr} \}_{\sfr=1}^{\sfR})$ without altering the distributions of $\bz$ and noise $\bn$.

\section{SC-VAMP-based decoder for SC-SS codes}
The minimum mean-square error (MMSE) estimator, which demonstrates excellent performance \cite{venkataramanan2019sparse}, is given by the 
expectation of the Bayes posterior
\begin{equation*}
    P(\bx|\by, \{\bA_{\sfr}\}_{\sfr}^{\sfR}) \propto \prod_{\sfr \leq \sfR} e^{-\frac{\snr}{2}\|\by_{\sfr} - \bA_{\sfr} \vec{\bx}_{\sfr}\|_2^2} \prod_{l \leq L} P_0(\bx_{\text{sec}(l)}).
\end{equation*}
The hard constraints for the sections of the message are enforced by the prior distribution $P_0(\bx_{\text{sec}(l)}) = B^{-1} \sum_{i \in \text{sec}(l)} \delta_{x_{i}, 1} \prod_{j \in \text{sec}(l), j \neq i} \delta_{x_{j}, 0}$.
However, direct calculation of the Bayes posterior expectation is intractable due to the high dimensional integration with time complexity $O(B^L)$. The proposed SC-VAMP algorithm (see Algorithm \ref{algo:SC-VAMP}) computes the MMSE estimator with polynomial time complexity.
\begin{algorithm}
\caption{SC-VAMP-based decoder for SC-SS codes}
\label{algo:SC-VAMP}
\setstretch{1.1}
\begin{algorithmic}[1]
\REQUIRE Max iteration $K$, design matrices $\{\bA_{\sfr}\}_{\sfr}^{\sfR}$, noised codeword $\by$
\STATE Initialize $\vec{\bp}_{1,\sfr}^0 = \bm{0}, \gamma_{1, \sfr}^0 = B$ for $\sfr \in [\sfR]$

\FOR{$k=0,1,\ldots, K$}
    \STATE // LMMSE estimation
    \FOR{$\sfr = 1, \ldots, \sfR$}
        \STATE $\vec{\bx}_{1, \sfr}^{k} = \bg_1(\vec{\bp}_{1, \sfr}^{k}, \gamma_{1,\sfr}^k, \by_{\sfr}, \bA_{\sfr} )$
        \STATE $\alpha_{1, \sfr}^k = \langle \bg_1^{\prime}(\vec{\bp}_{1, \sfr}^{k}, \gamma_{1,\sfr}^k, \by_{\sfr}, \bA_{\sfr}) \rangle$
        \STATE $\eta_{1, \sfr}^{k} = \gamma_{1, \sfr}^{k}  / \alpha_{1, \sfr}^{k}$
        \STATE $\gamma_{2, \sfr}^{k} = \eta_{1, \sfr}^{k} - \gamma_{1, \sfr}^{k}$
        \STATE $\vec{\bp}_{2, \sfr}^{k} = (\eta_{1,\sfr}^k \vec{\bx}_{1, \sfr}^k - \gamma_{1,\sfr}^k \vec{\bp}_{1,\sfr}^{k})/ \gamma_{2, \sfr}^{k}$
    \ENDFOR     
    \STATE // Denoising
    \FOR{$\sfc = 1, \ldots, \sfC$}
        \STATE $\hat{\gamma}_{2,\sfc}^{k} = \sum_{\sfr = \sfc}^{\sfc+W-1} \gamma_{2,\sfr}^k$ 
        \STATE $\hat{\bp}_{2,\sfc}^{ k} = (\sum_{\sfr = \sfc}^{\sfc+W-1} \gamma_{2, \sfr}^k \vec{\bp}_{2, \sfr}^{k}[\sfc + N(\sfr) -\sfr])/ \hat{\gamma}_{2,\sfc}^{k}$
        \STATE $\left(N(\sfr) = |W_{\sfr}| \II \{\sfr \leq W\} + W \II \{\sfr>W\} \right)$
        
        \STATE $\hat{\bx}^{ k}_{2, \sfc} = \bg_2(\hat{\bp}_{2,\sfc}^{k}, \hat{\gamma}_{2,\sfc}^{k})$
        \STATE $\hat{\alpha}_{2,\sfc}^{k} = \langle \bg_2^{\prime}(\hat{\bp}_{2,\sfc}^{k}, \hat{\gamma}_{2,\sfc}^{k})\rangle$
        \STATE $\hat{\eta}^{k}_{2,\sfc}= \hat{\gamma}_{2, \sfc}^{k} / \hat{\alpha}_{2, \sfc}^{k}$
    \ENDFOR
    \STATE // Concatenating
    \FOR{$\sfr = 1, \ldots, \sfR$}
        \STATE $\vec{\bx}_{2, \sfr}^{k} = \text{concate}(\{\hat{\bx}_{2, \sfc}^{ k}| c \in W_{\sfr}\})$ 
        \STATE $\eta_{2, \sfr}^{k} =|W_{\sfr}| / \left[\sum_{c \in W_{\sfr}} (\hat{\eta}_{2, \sfc}^{ k})^{-1} \right] $
        \STATE $\gamma_{1, \sfr}^{k+1} = \eta_{2, \sfr}^{k} - \gamma_{2, \sfr}^k$
        \STATE $\vec{\bp}_{1, \sfr}^{k+1} = (\eta_{2,\sfr}^k \vec{\bx}_{2, \sfr}^k - \gamma_{2, \sfr}^k \vec{\bp}_{2, \sfr}^k) / \gamma_{1, \sfr}^{k+1}$    
    \ENDFOR
    
\ENDFOR
\RETURN $\{ \tilde{\bx}_{\sfc}^K \}_{\sfc = 1}^{\sfC}$ (we define $\tilde{\bx}_{\sfc}^k = \hat{\bx}_{2,\sfc}^{k}$ for $\sfc \in [\sfC]$)
\end{algorithmic}
\end{algorithm}

VAMP was originally derived for generalized linear estimation \cite{rangan2019vector,ma2017orthogonal}, and subsequently extended with a spatial coupling structure to achieve the information-theoretic compressing limitation \cite{takeuchi2023orthogonal}. We utilize the factor graph shown in Fig.~\ref{fig:SC_factor}, based on the  EC algorithm  to derive the SC-VAMP decoder for SS code. The primary distinction from the factor graph in \cite{rangan2019vector} is that the variable factor $\bx_{\sfc}$ is connected to multiple function factors $\delta_{\sfr}$. To address this modification, we employ uniform diagonalization \cite{fletcher2016expectation}, i.e., replacing the original message $\mathcal{N}\left(\vec{\bx};\vec{\bx}_{2, \sfr}^{k}, \text{diag}(\{ (\hat{\eta}_{2, \sfc}^k)^{-1} \bI | \sfc \in W_{\sfr}\})\right)$ passed from the variable factors $\bx_{\sfc}$ to the function factor $\delta_{\sfr}$ with the approximation 
$\mathcal{N}(\vec{\bx};\vec{\bx}_{2, \sfr}^{k}, (\eta_{2,\sfr}^{k})^{-1} \bI_{N_{\sfr}})$, where $\eta_{2, \sfr}^{k} =|W_{\sfr}| / \left[\sum_{c \in W_{\sfr}} (\hat{\eta}_{2, \sfc}^{ k})^{-1} \right] $. This modification leads to the concatenating part in Algorithm \ref{algo:SC-VAMP}.

With message passing rules \cite{rangan2019vector} and uniform diagonalization, we derive the SC-VAMP for SS codes where the only difference from canonical SC-VAMP is that the so-called nonlinear denoiser $\bg_2(\bp, \gamma)$ acts section-wise rather than component wise. In full generality, the denoiser is defined as $\bg_2(\bp, \gamma) = \EE[\bX|\bP = \bp]$ for the random variable $\bP = \bX + \sqrt{\gamma}^{-1} \bZ$ with $\bX \sim P_0^{\otimes (L/\sfC)}$ and $\bZ \sim \cN(0, \bI_{N/\sfC})$. Plugging $P_0$ yields a component-wise expression for the denoiser and its variance:
\begin{equation*}
\begin{cases}
\left[\bg_{2}(\bp, \gamma)\right]_{i}=\frac{\exp \left( \gamma p_{i} \right)}{\sum_{j \in \text{sec}(l_{i})} \exp \left(\gamma p_{j}  \right)}, \\
\left[\bg_{2}^{\prime}(\bp, \gamma)\right]_{i}=\gamma\left[\bg_{2}(\bp, \gamma)\right]_{i}\left(1-\left[\bg_{2}(\bp, \gamma)\right]_{i}\right),
\end{cases}
\end{equation*}
where $\left[ \bg_2^{\prime}(\bp, \gamma)\right]_{i}:= \partial_{p_i} [\bg_2(\bp, \gamma)]_i$, $l_i$ is the index of the section to which the $i$-th scalar component belongs. $\bg_2(\vec{\bp}, \gamma, \by, \bA)$ can be interpreted as the MMSE estimate of the random vector $\vec{\bx}$ given the observation $\by \in \RR^{m}$ and the matrix  $\bA \in \RR^{m \times n}$, with $\by \sim \mathcal{N}(\bA \vec{\bx}, \snr^{-1} \bI_{n})$ and the prior  is $\vec{\bx} \sim \cN(\vec{\bp}, \gamma^{-1} \bI_{n})$:
\begin{equation*}
\begin{cases}
\bg_{1}(\vec{\bp}, \gamma, \by, \bA) =\left(\operatorname{snr} \bA^{\mathsf{T}} \bA+\gamma \bI_{n}\right)^{-1}\left( \operatorname{snr} \bA^{\mathsf{T}} \by+\gamma \vec{\bp}\right), \\
\langle\bg_{1}^{\prime}(\vec{\bp}, \gamma, \by, \bA)\rangle =\gamma n^{-1} \operatorname{Tr}\left[\left(\operatorname{snr} \bA^{\mathsf{T}} \bA+\gamma \bI_n\right)^{-1}\right].
\end{cases}
\end{equation*}

The primary computational cost in the SC-VAMP decoder stems from the matrix inversion in the estimator $\bg_2$, which results in a time complexity of $O(B^3L^3)$ per iteration. However, if all singular values of $\bA$ are equal to $\sqrt{a}$, $\bg_1$ can be simplified as follows:
\begin{equation*}
    \bg_1(\vec{\bp}, \gamma, \by, \bA) = \vec{\bp} + \frac{\snr}{\gamma+a\snr} \bA^{\mathsf{T}}(\by - \bA \vec{\bp}).
\end{equation*}
In this case, the main computational cost is reduced to matrix-vector multiplication.
Furthermore, if the design matrices $\{{\bA}\}_{\sfr=1}^{\sfR}$ are generated from structured matrices, such as
DCT matrices or 
Hadamard matrices \cite{abbara2020universality}, the overall time complexity per iteration is reduced to $O((BL) \log (BL))$.

\begin{figure}
    \centering
    \includegraphics[width=1.05\linewidth]{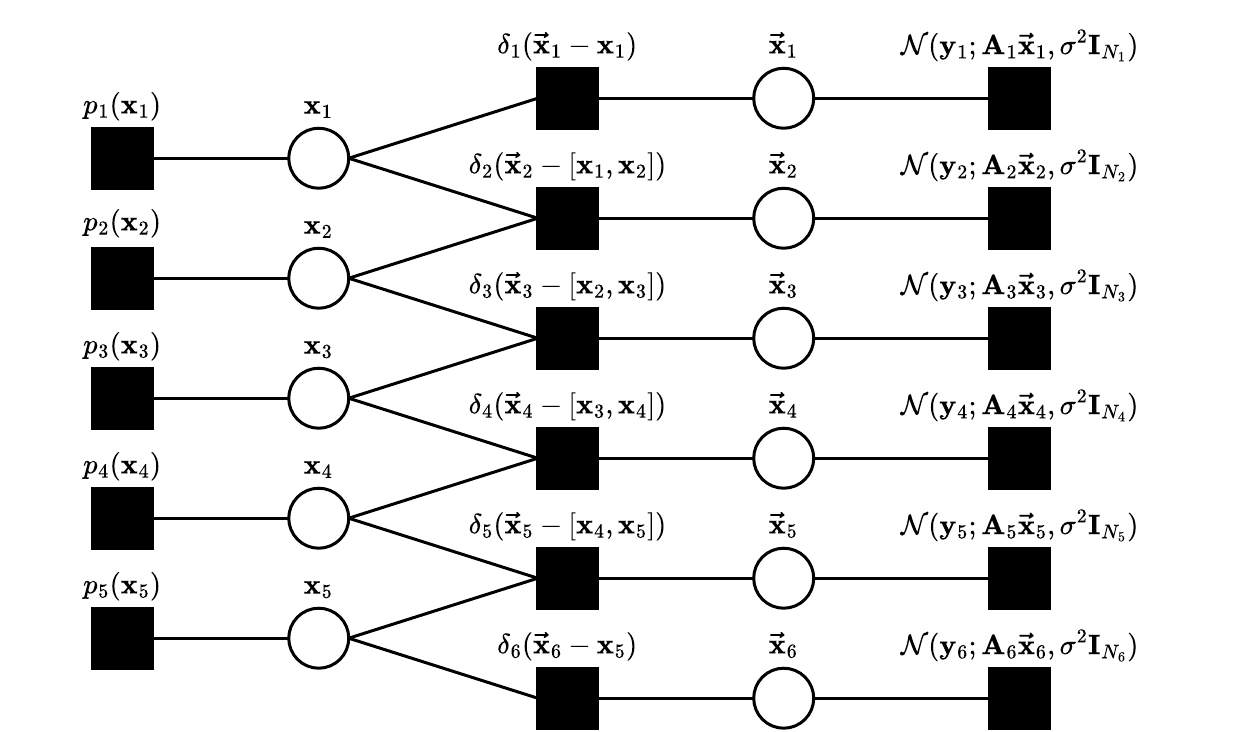}
    \vspace{-5mm}
    \caption{Factor graph  representation of the joint probability distribution of  $\{\bx_{\sfc}\}_{\sfc}^{\sfC}$ and $\{ \vec{\bx}_{\sfr}\}_{\sfr}^{\sfR}$, with spatial coupling parameters $\Gamma = 5, W=2$ and $\vartheta = 6/5$. The function factors $\{ p_{\sfc}(\cdot)\}_{\sfc}^{\sfC}$ represent $P_0^{\otimes(L/\sfC)}$, and the factors
    $\{ \delta_{\sfr}(\cdot)\}_{\sfr}^{\sfR}$  represent the Dirac delta distributions.}
    \label{fig:SC_factor}
    \vspace{-5mm}
\end{figure}


\section{State evolution analysis}
To evaluate the performance of the SC-VAMP estimator $\tilde{\bx}_{\sfc} = (\tilde{\bx}_{l_{\sfc}}, \ldots, \tilde{\bx}_{r_{\sfc}})$, we define two key error metrics for each block: the mean squared error (MSE) per section, denoted by $\{ E_{\sfc}^k\}_{\sfc=1}^{\sfC}$ and the section error rate (SER), denoted by $\{ \text{SER}_{\sfc}^{k}\}_{\sfc=1}^{\sfC}$. These metrics are expressed as follows:
\begin{equation*}
    E_{\sfc}^k = \frac{\sfC}{L}\| \tilde{\bx}_{\sfc}^k - \bx_{\sfc}\|_2^2, \quad \text{SER}_{\sfc}^{k} = \frac{\sfC}{L} \sum_{l=l_{\sfc}}^{r_{\sfc}} \II\{\tilde{\bx}_{\text{sec}(l)} \neq {\bx}_{\text{sec}(l)}\},
\end{equation*}
for $\sfc \in [\sfC]$, where $\II \{\cdot\}$ is the indicator function. In the asymptotic $L \to \infty$ limit (with $\alpha, B, R_{\text{all}}$ fixed), the MSE metric can be tracked by the following SE equations
\begin{align*}
    \bar{\eta}_{1, \sfr}^{k} &= [\cE_{1,\sfr}(\bar{\gamma}_{1,\sfr}^k)]^{-1},\quad \bar{\gamma}_{2,\sfr}^k = \bar{\eta}_{1, \sfr}^k - \bar{\gamma}_{1, \sfr}^k, \quad \sfr \in [\sfR],  \\
    \bar{\hat{\gamma}}_{2, \sfc}^k &= \sum_{\sfr = \sfc}^{\sfc+W-1} \bar{\gamma}_{2,\sfr}^k, \quad \bar{\hat{\eta}}_{2, \sfc}^k = [\cE_2(\bar{\hat{\gamma}}_{2, \sfc}^k)/B]^{-1},  \quad \sfc \in [\sfC],  \\
     \bar{\eta}_{2,\sfr}^k &= |W_{\sfr}| / \sum_{\sfc \in  W_{\sfr}} (\bar{\hat{\eta}}_{2, \sfc}^k)^{-1}, \quad \bar{\gamma}_{1, \sfr}^{k+1} = \bar{\eta}_{2, \sfr}^k - \bar{\gamma}_{2, \sfr}^k,   \quad \sfr \in [\sfR],
\end{align*}
with initialization $\bar{\gamma}_{1, \sfr}^0 = B$ for $\sfr \in [\sfR]$. Here, $\cE_{1, \sfr}(\gamma) = \EE_{\rho_{\sfr}} (\snr B\lambda+\gamma)^{-1}$ and 
\begin{equation*}
    \cE_2(\gamma) = 1 - \mathbb{E}\left[\frac{e^{  \sqrt{\gamma}U_{1}}}{e^{ \sqrt{\gamma}U_{1}  }+e^{-\gamma} \sum_{j=2}^{B} e^{\sqrt{\gamma}U_{j}  }}\right]
\end{equation*}
with $U_{1}, \ldots, U_{B} \stackrel{\text { i.i.d. }}{\sim} \mathcal{N}(0,1)$. The error function $\cE_2(\gamma)$ can be interpreted as the Bayes posterior variance, i.e., $\cE_2(\gamma) = \EE \| \bS - \EE[\bS | \bS + \sqrt{\gamma}^{-1} \bZ] \|_2^2$ with $\bS \sim P_0$ and $\bZ \sim \mathcal{N}(0, \bI_B)$. With initialization $\vec{p}_{1, \sfr}^{0} = 0, \gamma_{1,\sfr}^0 = \bar{\gamma}_{1, \sfr}^0 = B$ for $\sfr \in [\sfR]$, it follows that for any finite $B$, $ \lim_{L \to \infty} E_{\sfc}^K = B( \bar{\hat{\eta}}_{2, \sfc}^K)^{-1}$ almost surely for $\sfc \in [\sfC]$. Next, we analyze SE in the limit of large section size, i.e., as $B \to \infty$ while keeping $R_{\text{all}}$ fixed (necessarily $\alpha = \Theta(\log B/ B) \to 0$). With initialization $\gamma_{1,\sfr}^0 = \Theta(B)$, the order of $\bar{\eta}_{1, \sfr}^{k}, \bar{\hat{\eta}}_{2, \sfc}^k, \bar{\eta}_{2, \sfr}^k, \bar{\gamma}_{1, \sfr}^k$ is $\Theta(B)$, while the order of $ \bar{\gamma}_{2, \sfr}^k$ and $\bar{\hat{\gamma}}_{2, \sfc}^k$ is $\Theta(\log B)$. Introducing the  rescaled variables $\sigma_{\sfr}^k =B / \bar{\gamma}_{1, \sfr}^k$, $\tau_{\sfc}^k =\log B /\bar{\hat{\gamma}}_{2, \sfc}^{k}$, $\psi_{\sfc}^{k} =B/ \bar{\hat{\eta}}_{2, \sfc}^k $, and $\phi_{\sfr}^{k} =  R_{\sfr} \bar{\gamma}_{2,\sfr}^k / \log B$, the limit SE equations are given by:
\begin{align}
    & \phi_{\sfr}^{k} = F_{\sfr}(\sigma_{\sfr}^k), \quad \sfr \in [\sfR], \label{eq: phi_r}\\
    & \tau_{\sfc}^{k} = R_{\text{all}}\left[\sum_{\sfr=\sfc}^{\sfc+W-1} \frac{\phi_{\sfr}^k}{\vartheta|W_{\sfr}| } \right]^{-1}, \quad \sfc \in [\sfC], \label{eq:tau_c}\\
    &\psi_{\sfc}^{k} = \II \{ \tau_{\sfc}^k > \frac{1}{2} \}, \quad \sfc \in [\sfC], \label{eq:psi_c}\\
    & \sigma_{\sfr}^{k+1} = \frac{1}{|W_{\sfr}|} \sum_{\sfc \in W_{\sfr}} \psi_{\sfc}^{k} , \quad \sfr \in [\sfR] \label{eq:sigma_r},
\end{align}
with the initialization $\sigma_{\sfr}^0 = 1$ for $\sfr \in [\sfR]$. Here, we exploit the phase transition of $\cE_2(\gamma)$ \cite{barbier2016proof,rush2021capacity}, where $\lim_{B \to \infty} \cE_2(\gamma) = \II \{ (\lim_{B\to\infty }  \log B /\gamma) > 1/2\}$ and $F_{\sfr}(x) = \EE_{\rho_{0, \sfr}} \frac{\lambda }{\lambda x + \sigma^2 }$  for $\sfr \in [\sfR]$, where $\rho_{0, \sfr}$ is the $\alpha_{\sfr} \to 0$ limit of $\rho_{\text{supp},\sfr}$. If all the limit to $\rho_{\text{supp}, \sfr}$ converges to the same $\rho_{0}$ (e.g., when $\bA_{\sfr}$ is a Gaussian matrix or $\rho_{\text{supp}, \sfr}$ is independent of $\alpha_{\sfr}$ and $B$ \cite{hou2022sparse,liu2022sparse}), we define $F_{\sfr}(x) = F(x) = \EE_{\rho_0} \frac{\lambda }{\lambda x+ \sigma^2}$. This allows us to track $E_{\sfc}^k$ with $\psi_{\sfc}^k$ in the large section limitation. In Proposition \ref{thm1}, we analyze the evolution of the limited SE.


\newtheorem{theorem}{Proposition}
\begin{theorem} 
    \label{thm1}
    Consider the limit SE with the same $F_{\sfr}(x)=F(x)$ and initialization  $\sigma^0_{\sfr} =1$ for $\sfr \in [\sfR]$. We define $R_{\textnormal{alg}} := F(1) / 2$ and $R_{\textnormal{IT}} := \int_{0}^{1} F(x) \mathrm{d} x / 2$. Then we have 
    \begin{itemize}
        \item If $R_{\textnormal{all}}<\vartheta^{-1} R_{\textnormal{alg}}$, then $\psi_{\sfc}^1 = 0$ for $\sfc \in [\sfC]$;
        \item If $\vartheta^{-1} R_{\textnormal{alg}}<R_{\textnormal{all}}<\vartheta^{-1} R_{\textnormal{IT}}$ and 
        \begin{equation} \label{eq:W}
            W>\max \{ \lceil 1/ l^*(\vartheta)\rceil, \lceil  1/h^*(\vartheta, \Delta)\rceil \},
        \end{equation}
        where $\Delta:=\vartheta^{-1} R_{\textnormal{IT}} - R_{\textnormal{all}}$, $l^*(\vartheta)$ is the unique solution of the equation in $(0,1)$: 
        \begin{equation} \label{eq:lx}
            l(x) = x-\ln x = \frac{\vartheta R_{\textnormal{all}}}{R_{\textnormal{alg}}}
        \end{equation}
        and $h^*(\vartheta,\Delta)$ is the unique solution of the equation in $(0,1)$:
        \begin{equation} \label{eq:hx}
            h(x) = \int_0^x F(t) \mathrm{d}t - F(1) x = 2 \vartheta \Delta.
        \end{equation}
         We have $\psi_{\sfc}^{k} = \psi_{\Gamma-\sfc+1}^k = 0$ for $1 \leq \sfc \leq \min\{(k+1)g, \lceil \Gamma / 2 \rceil\}$ where $g := \min \{ \lfloor h^*(\vartheta, \Delta) W \rfloor,  \lfloor l^*(\vartheta)     W\rfloor\}$.
        Furthermore, after $K = 1 + \lceil \Gamma / (2g) \rceil$ iterations, we have $\psi_{\sfc}^K = 0$ for $\sfc \in [\sfC]$.
    \end{itemize}    
\end{theorem}
\begin{proof}
    Since the variables $\psi_{\sfc}^{k}$ for $\sfc \in [\Gamma]$ are symmetric about the center column index, i.e. $\psi_{\sfc} = \psi_{\Gamma - \sfc +1}$ for $\sfc <\lceil \Gamma / 2\rceil$, we carry out the analysis for $\sfc \in [\lceil \Gamma/2 \rceil]$ and $\sfr \in [\lceil \Gamma/2 \rceil+W-1]$. With initialization $\sigma_{\sfr}^0=1$ for $\sfr \in [\sfR]$, we have 
    \begin{equation*}
        \tau_{\sfc}^0 = \frac{R_{\textnormal{all}}\vartheta}{F(1)} 
        \begin{cases}
            \left[\sum_{\sfr = \sfc}^{W-1} \frac{1}{\sfr}+ \frac{\sfc}{W}\right]^{-1},  &1\leq \sfc \leq W, \\
            1, &\sfc>W.
        \end{cases}
    \end{equation*}
    Here, $\tau_{\sfc}^0$ is non-decreasing with the block index $\sfc$. Thus, if $R_{\textnormal{all}}< \vartheta^{-1} R_{\textnormal{alg}}$, we have $\tau_{\sfc}^0 \leq \frac{1}{2}$ implying $\psi_{\sfc}^{1} = 0$ for $\sfc \in [\sfC]$.

    Otherwise, if $\vartheta^{-1} R_{\textnormal{alg}}<R_{\textnormal{all}}<\vartheta^{-1} R_{\textnormal{IT}}$. Firstly, we establish that $l(x)$ and $h(x)$ each have a unique solution on $(0,1)$, respectively. Differentiating $l(x)$ over $(0,1)$, we obtain $l^{\prime}(x) = 1 - 1/x <0$, implying $l(x)$ is decreasing for $0<x<1$. Given that $\lim_{x\to0^{+}} l(x) = +\infty$ and $l(1) = 0< (\vartheta R_{\textnormal{all}}) /R_{\textnormal{alg}}$ with continuity of $l(x)$ in $(0,1]$, it follows that equation (\ref{eq:lx}) has a unique solution $l^*(\vartheta)$ on $(0,1)$. We differentiate $h(x)$ over $[0,1]$ to obtain
    \begin{equation*}
        h^{\prime}(x) = F(x) - F(1)
    \end{equation*}
    and since $F(x)$ is decreasing on $(0,1)$, it follows that $h^{\prime}(x)\geq 0$, implying $h(x)$ is increasing on $(0,1)$. Moreover, with $h(0) = 0 < 2 \vartheta \Delta $, $h(1) = 2(R_{\textnormal{IT}}-R_{\textnormal{alg}}) > 2 \vartheta \Delta$ and the continuity of $h(x)$ on $[0,1]$, we conclude that equation (\ref{eq:hx}) has a unique solution $h^*(\vartheta, \Delta)$ in $(0,1)$.

    With the inequality $\sum_{\sfr=\sfc}^{W-1} (1/\sfr) \leq \ln(W / \sfc)$ for $1 \leq \sfc \leq W$, we deduce the following inequality of $\tau_{\sfc}^0$:
    \begin{equation*}
        \tau_{\sfc}^0 \leq \frac{R_{\textnormal{all}}\vartheta}{F(1)} 
        \begin{cases}
            \left[\frac{\sfc}{W} -\ln (\frac{\sfc}{W})\right]^{-1}, & 1 \leq \sfc \leq W, \\
            1, &\sfc > W.
        \end{cases}
    \end{equation*} 
    Since $l(x)$ is decreasing on $(0,1)$, we have $\tau_{\sfc}^0< 1/ 2$, implying $\psi_{\sfc}^0 = 0$ for $1 \leq \sfc \leq g$ ($g$ is a integer greater than or equal to $1$ according to \eqref{eq:W}. 

    Next we consider the subsequent iterations $k \geq 1$. Assume towards induction that $\psi_{\sfc}^k = 0$ for $\sfc \leq g^k$ where $g^k \geq (k+1) g$. Based on the assumption, we deduce the inequality for $\sigma_{\sfr}^{k+1}$ as follows:
    \vspace{-0.5em}
    \begin{equation*}
        \sigma_{\sfr}^{k+1} \geq 
        \begin{cases}
            0, & \sfr \leq g^{k}, \\
            (\sfr-g^{k})/|W_{\sfr}|, &g^k \leq \sfr \leq g^{k}+W, \\
            1, & \sfr > g^k+W.
        \end{cases}
    \end{equation*}
    Using the fact $F(x)$ is decreasing on $[0,1]$ and (\ref{eq: phi_r}), we obtain
    \begin{equation*}
        \phi_{\sfr}^{k+1} \geq 
        \begin{cases}
            F(0), & \sfr \leq g^{k}, \\
            F((\sfr-g^{k})/|W_{\sfr}|), &g^k \leq \sfr \leq g^{k}+W, \\
            F(1), & \sfr > g^k+W.
        \end{cases}
        \vspace{-0.5em}
    \end{equation*}
    
    Since $\psi_{\sfc}^k$ increases with the block index $\sfc \in [\lceil \Gamma / 2\rceil]$, we have $\sigma_{\sfr}^{k+1}$ increases with the block index $\sfr$ for $\sfr \in [\lceil \Gamma/2 \rceil+W-1]$ according to (\ref{eq:sigma_r}).
    Given that $F(x)$ is decreasing on $[0,1]$, we conclude that $\phi_{\sfr}^{k+1} / |W_{\sfr}| = F(\sigma_{\sfr}^{k+1}) / |W_{\sfr}|$ decreases with the block index $\sfr$ for $\sfr \in [\lceil \Gamma/2 \rceil+W-1]$ according to (\ref{eq: phi_r}). Therefore, $\tau_{\sfc}^{k+1}$ increases with the block index $\sfc$ for $\sfc \in \lceil \Gamma / 2\rceil$ according to (\ref{eq:tau_c}). We now consider $\tau_{\sfc}^{k+1}$ for $g^k<\sfc \leq g^k + W$ as follows:
    \begin{equation*} 
        \frac{\tau_{\sfc}^{k+1}}{R_{\textnormal{all}} \vartheta} \leq  \Bigg[\sum_{\sfr=\sfc}^{g^k+W-1} \!\! \dfrac{F((\sfr-g^k)/|W_{\sfr}|)}{|W_{\sfr}|} + \!\!\sum_{\sfr =g^k+W}^{\sfc+W-1} \! \dfrac{F(1)}{|W_{\sfr}|} \Bigg]^{-1}.
    \end{equation*}
    Since $xF(bx)$ is increasing on $[0,1]$ for $b>0$ and $|W_{\sfr}| \leq W$ for $\sfr \in [\sfR]$, we obtain the following inequality:
    \begin{equation*} 
        \frac{\tau_{\sfc}^{k+1}}{R_{\textnormal{all}} \vartheta} \leq  \Bigg[\sum_{\sfr=\sfc}^{g^k+W-1} \!\! \dfrac{F((\sfr-g^k)/|W|)}{|W|} + \!\!\sum_{\sfr =g^k+W}^{\sfc+W-1} \! \dfrac{F(1)}{|W|} \Bigg]^{-1}.
    \end{equation*}
    Given that $F(x)$ is decreasing in $[0,1]$, we have
    \begin{equation*}
        \tau_{\sfc}^{k+1} \leq R_{\textnormal{all}} \vartheta \Bigg[ \int_{\frac{\sfc-g^k}{W}}^1 F(x) \mathrm{d} x + \frac{\sfc-g^k}{W} F(1)\Bigg]^{-1}.
    \end{equation*}
    Finally, since $h(x)$ is decreasing on $(0,1)$, it follows that  $\tau_{\sfc}^{k+1} \leq \frac{1}{2}$, implying $\psi_{\sfc}^{k+1} = 0$ for $\sfc \leq g^{k} +g$. Moreover, since $Kg \geq \lceil \Gamma / 2\rceil$, it follows that $\psi_{\sfc}^K = 0 $ for $\sfc \in [\sfC]$.    
\end{proof}


\begin{figure}[t]
    \centering
    \includegraphics[width=1\linewidth]{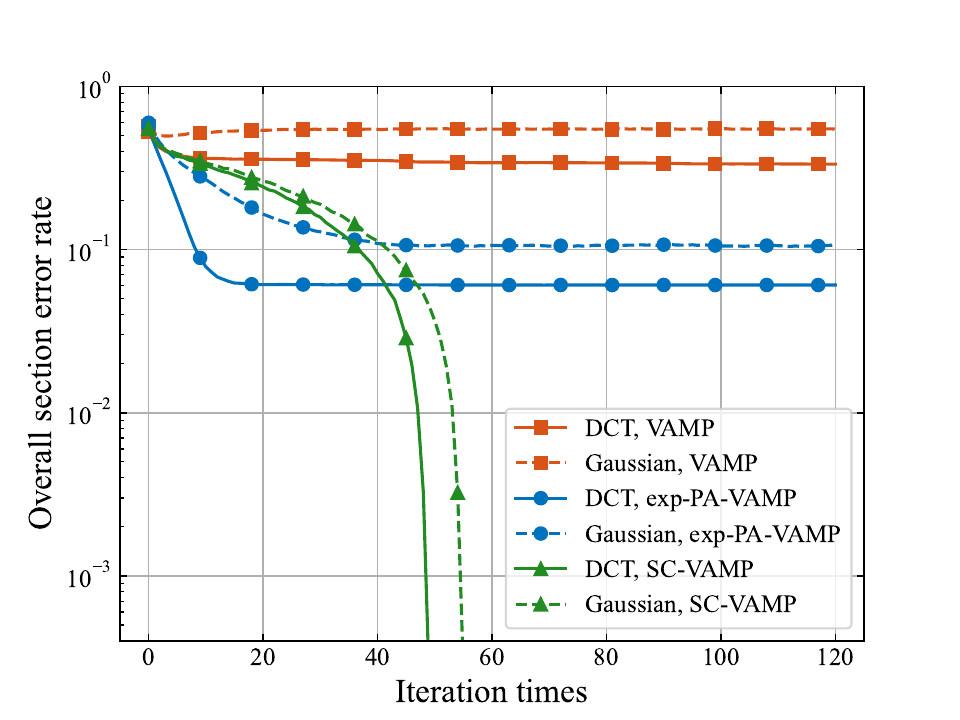}
    \vspace{-5mm}
    \caption{Overall section error rate $\text{SER}^k = (\sum_{\sfc=1}^\sfC \text{SER}^{k}_{\sfc} )/ \sfC$ ran on a random single instance with $L = 2^{14}$, $B=16$ and $\snr = 15$, as a function of the iterations. For SC-VAMP decoder, we set $\Gamma = 16$, $W = 2$ and $\vartheta = 17/16$. The overall rate $R_{\text{all}} = 1.60$ is higher than the algorithm threshold, lower than the information-theoretic threshold proposed in \cite{hou2022sparse}.}
    \label{fig:decoder-SER}
    \vspace{-5mm}
\end{figure}

The Proposition $\ref{thm1}$ establishes that, under $R_{\textnormal{all}} < \vartheta^{-1} R_{\textnormal{IT}}$,
the limit SE iteratively transitions $\psi_{\sfc}^k$ from $1$ to $0$, progressing from the outermost blocks toward the central blocks. And after $K$ iterations, $\psi_{\sfc}^K$ reduces to $0$ for all blocks. Moreover, for sufficiently large $W$ with $\Gamma > W^2$, we have $\vartheta \to 1$ as $W \to \infty$. Consequently, Proposition \ref{thm1} implies that perfect recovery of the message $\bx$ is achievable when $R_{\textnormal{all}} < R_{\textnormal{IT}}$, thereby confirming the the conjectures in \cite{liu2022sparse}.

Using Cauchy–Schwarz inequality, we obtain 
\vspace{-1.1mm}
\begin{equation*}
    \EE_{\rho_0} \left(\frac{\lambda}{\lambda x + \sigma^2}\right) \EE_{\rho_0}(\lambda x + \sigma^2) \leq \EE_{\rho_0} \lambda,
    \vspace{-0.7mm}
\end{equation*}
where the equality holds if and only if $\rho_0 = \delta_1$. This implies that $F(x) \leq 1/(\sigma^2 + x)$, leading to the following bound:
\vspace{-1.1mm}
\begin{equation*}
    R_{\textnormal{IT}} \leq \frac{1}{2}\int_0^1 \frac{1}{\sigma^2 + x} \mathrm{d} x = \frac{1}{2} \log(1 + \text{snr}) = \mathcal{C}. 
    \vspace{-0.7mm}
\end{equation*}
Therefore, the SC-SS code with SC-VAMP decoder is capacity-achieving for AWGNC if all $\rho_{\text{supp}, \sfr}$ converges to $\delta_1$ as $\alpha \to 0$, a condition referred to as the spectra criterion in 
\cite{hou2022sparse,liu2022sparse,xu2023capacity}. 
\section{Numerical results}
We present our numerical results using the following two design matrices: 
\begin{enumerate} 
    \item \textbf{Gaussian matrices.} Each element of $\bA_{\sfr}$ is \iid drawn from $\cN(0, L^{-1})$. The limit spectral distribution $\rho_{\sfr}$ follows the Marchenko–Pastur law, given by $\rho_{\sfr}(\lambda) = (1 - \alpha_{\sfr})\delta_{\lambda,0} + \frac{\sqrt{\left(\lambda-\lambda_{\sfr}^-\right)\left(\lambda_{\sfr}^+-\lambda\right)}}{2 \pi \lambda}$, where $\lambda^{ \pm}_{\sfr}:=(1 \pm \sqrt{\alpha_{\sfr}})^{2}$. 
    \item \textbf{DCT matrices.} To simulate row-orthogonal matrices with the limiting spectral distribution $\rho_{\sfr}(\lambda) = (1 - \alpha_r) \delta_{\lambda,0} + \alpha_{\sfr} \delta_{\lambda,1}$, we randomly choose $M_{\sfr} \leq N_{\sfr}$ rows from a $N_{\sfr} \times N_{\sfr}$ DCT matrix.
\end{enumerate}

In Fig.~\ref{fig:decoder-SER}, we present a comparison between our SC-VAMP decoder, the original VAMP decoder \cite{hou2022sparse}, 
and the exp-PA-VAMP decoder, which incorporates exponential decay power allocation \cite{xu2023capacity}. Our results show that the original VAMP decoder fails to decode successfully because $R_{\text{all}}$ exceeds the algorithmic threshold. In contrast, both the SC-VAMP and the exp-PA-VAMP decoders achieve successful decoding. Notably, the SC-VAMP decoder demonstrates superior empirical performance, achieving an overall section error rate of less than $10^{-3}$. Furthermore, we observe that DCT matrices consistently outperform Gaussian matrices.

Figure~\ref{fig:block_SER} illustrates the dynamics of the section error rate in the SC-VAMP decoding process. We observe that the SC-VAMP decoder successfully decodes progressively from the outermost blocks towards the central blocks, thereby confirming the validity of Proposition \ref{thm1}.

\begin{figure}[t]
    \centering
    \includegraphics[width=1\linewidth]{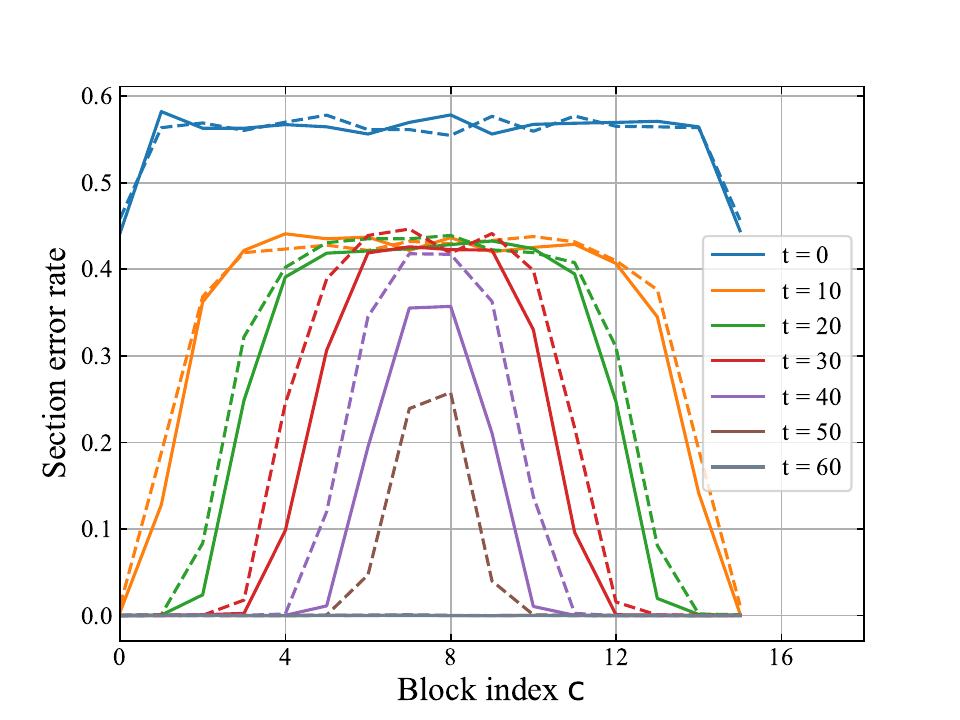}
    \vspace{-5mm}
    \caption{$\text{SER}_{\sfc}^k$ vs. block index $\sfc \in [\sfC]$ for several iteration numbers in a SC-VAMP decoding process, using the same setup as in Fig.\ref{fig:decoder-SER}. The solid line corresponds to DCT matrices, while the dotted line corresponds to Gaussian matrices.}
    \label{fig:block_SER}
    \vspace{-5mm}
\end{figure}
\vspace{-1mm}
\section{Further Work}
\vspace{-1mm}
In practice, we aim to derive non-asymptotic section error rate bounds for
$B, L \to \infty$ with appropriate scaling, e.g., $B = O(L^{\alpha})$
and expect exponential decay of these bounds using techniques from
\cite{rush2021capacity,rush2018finite,li2022non,cademartori2024non}. 
We also expect these results to extend to general memoryless channels, as suggested informally in
\cite{liu2022sparse} from a statistical physics perspective. Further open problems related to SC-SS include designing faster, more robust decoders based on message-passing algorithms in \cite{liu2022memory,fan2022approximate}, and analyzing the performance of SC-SS with more realistic design matrices, such as the class of semi-random matrices \cite{dudeja2023universality,dudeja2024spectral}.

\clearpage
\bibliographystyle{IEEEtran.bst}
\bibliography{ref.bib}

\end{document}